\documentclass[a4p, 12pt]{article}
\title{Nonexistence of spontaneous symmetry breakdown  
of time-translation symmetry on general quantum systems:\\
Any macroscopic order parameter  moves not!}
\date{Aug   2023}
\author{Hajime Moriya}
\author{Hajime Moriya\thanks
{Faculty of Mechanical Engineering, Institute of Science and Engineering, Kanazawa University,
Kakuma, Kanazawa 920-1192, Japan.}}  
\usepackage{latexsym}
\usepackage{amssymb}
\usepackage{amsmath} %
\usepackage{amsthm} %
\usepackage{amscd} %
\usepackage{mathrsfs} 
\newtheorem{theorem}{Theorem}[section]
\newtheorem{corollary}[theorem]{Corollary}

\newtheorem{prop}[theorem]{Proposition}
\theoremstyle{remark}
\numberwithin{equation}{section}
\newtheorem{remark}{Remark}
\begin{document}
\maketitle

\begin{abstract}
The   Kubo-Martin-Schwinger (KMS) condition 
 is a well-founded general  definition 
 of  equilibrium states on  quantum systems. 
  The  time invariance property  of  equilibrium states  
 is  one  of  its  basic consequences. 
 From  the time invariance of any equilibrium state 
 it follows that the  spontaneous  breakdown  of  time-translation  symmetry  
 is impossible. 
  Moreover,   triviality of  the temporal long-rang-order 
  is  derived   from  the KMS condition. 
Therefore,  
the manifestation    of genuine quantum time crystals 
 is impossible  as long as  
  the standard  notion   of spontaneous symmetry  breakdown is considered.
\end{abstract}
\newcommand{\cstar}{{C}^{\ast}}%
\newcommand{\wstar}{{W}^{\ast}}%
\newcommand{\Al}{{\cal{A}}}%
\newcommand{\AUT}{{\rm{Aut}}(\Al)}%
\newcommand{\Fin}{{\mathfrak F}}%
\newcommand{\Finf}{\Fin_{{\rm{loc}} }}%
\newcommand{\R}{{\mathbb{R}}}%
\newcommand{\Z}{{\mathbb{Z}}}%
\newcommand{\CC}{{\mathbb{C}}}%
\newcommand{\NN}{{\mathbb{N}}}%
\newcommand{\nonum}{\nonumber}%
\newcommand{\Zmu}{\Z^{\mu}}%

\newcommand{\al}{\alpha}
\newcommand{\alz}{\alpha_{z}}%
\newcommand{\ome}{\omega}%
\newcommand{\vp}{\varphi}
\newcommand{\lam}{\lambda}%
\newcommand{\Lam}{\Lambda}%
\newcommand{\bLam}{\bar{\Lambda}}%
\newcommand{\Lamc}{\Lam^{c}}%
\newcommand{\vpweak}{\tilde{\vp}}
\newcommand{\Lamp}{\Lam^{\prime}}%
\newcommand{\Lami}{\Lambda_1}%
\newcommand{\Lamii}{\Lambda_2}%
\newcommand{\Lamiii}{\Lambda_3}%
\newcommand{\tr}{{\rm{tr}}}%
\newcommand{\I}{{\mathrm{I}}}%
\newcommand{\Ix}{\I_x}%
\newcommand{\J}{{\mathrm{J}}}%
\newcommand{\K}{{\mathrm{K}}}%
\newcommand{\LL}{{\mathrm{L}}}%
\newcommand{\vome}{\varomega}
\newcommand{\Ome}{\Omega}%
\newcommand{\ometil}{\widetilde{\ome}}%
\newcommand{\vptil}{\widetilde{\vp}}%
\newcommand{\alvp}{\widetilde{\alpha}_{\vp}}%
\newcommand{\alvpt}{\widetilde{\alpha}_{\vp, t}}%
\newcommand{\All}{\Al_{\Lam}}%
\newcommand{\Allp}{\Al_{\Lamp}}%
\newcommand{\Allx}{\Al_{\Lam+x}}
\newcommand{\Alloc}{\Al_{\rm{loc}}}%
\newcommand{\Alx}{\Al_{\{x\}}}%
\newcommand{\Alzero}{\Al_{\{0\}}}%
\newcommand{\vNM}{{\mathfrak{M}}}%
\newcommand{\ZZ}{{\mathfrak Z}}%
\newcommand{\Bl}{{\mathfrak B}}%
\newcommand{\Imag}{{\bf{Im}}}%
\newcommand{\Imrm}{{\rm{Im}}}%
\newcommand{\vno}{\vNM_{\ome}}%
\newcommand{\vnv}{\vNM_{\vp}}%
\newcommand{\vnome}{\vNM_{\ome}}%
\newcommand{\State}{S(\Al)}%
\newcommand{\FACState}{S_{\rm{factor}}(\Al)}%

\newcommand{\GState}{S^{\theta, G}_{\rm{inv}}(\Al)}%
\newcommand{\STATIONState}{S^{\alpha, \R}_{\rm{inv}}(\Al)}%
\newcommand{\STATSTATE}{S_{\alpha_t}^{\rm{stat.}}(\Al)}%
\newcommand{\StateLam}{S(\All)}
\newcommand{\TRASTATE}{S^{\tau, \Gamma}_{\rm{inv}}(\Al)}%
\newcommand{\HOMOSTATE}{S^{\tau}_{\rm{homo.}}(\Al)}%
\newcommand{\perSTATE}{S^{\tau, \Delta}_{\rm{inv}}(\Al)}%
\newcommand{\perprimeSTATE}{S^{\tau,\Delta^{\prime}}_{\rm{inv}}(\Al)}%
\newcommand{\setKMSb}{S_{\alpha_t, \beta}(\Al)}%
\newcommand{\setACCUM}{S^{\lim}_{\rm{Gibbs}, \beta}(\Al)}%
\newcommand{\EXsetKMSb}{S^{\rm{ext}}_{\alpha_t, \beta}(\Al)}%
\newcommand{\setGround}{S_{\alpha_t, \infty}(\Al)}%
\newcommand{\setEQU}{S_{\alpha_t}^{\rm{Equil}}(\Al)}%
\newcommand{\Hil}{{\cal H}}%
\newcommand{\Hilvp}{{\Hil}_{\vp}}%
\newcommand{\pivp}{\pi_{\vp}}%
\newcommand{\Omevp}{\Ome_{\vp}}%
\newcommand{\Hilome}{{\Hil}_{\ome}}%
\newcommand{\piome}{\pi_{\ome}}%
\newcommand{\Omeome}{\Ome_{\ome}}%
\newcommand{\Hilpsi}{{\Hil}_{\psi}}%
\newcommand{\pipsi}{\pi_{\psi}}%
\newcommand{\Omepsi}{\Ome_{\psi}}%
\newcommand{\vnvppri}{\vNM_{\vp}^{\,\prime}}
\newcommand{\vnomepri}{\vNM_{\ome}^{\,\prime}}
\newcommand{\vnpsipri}{\vNM_{\psi}^{\,\prime}}
\newcommand{\Dbeta}{D_{\beta}}%
\newcommand{\Dbetac}{\overline{D_{\beta}}}%
\newcommand{\Dbetao}{\stackrel{\circ}{D}_{\beta} }%
\newcommand{\Lamn}{\Lam_0(n)}
\newcommand{\tilLamn}{\tilde{\Lam}_0(n)}
\newcommand{\Lamr}{\Lam_0(r)}
\newcommand{\Lamm}{\Lam_0(m)}
\newcommand{\AlLamn}{\Al_{\Lamn}}
\newcommand{\AlLamr}{\Al_{\Lamr}}
\newcommand{\AlLamm}{\Al_{\Lamm}}
\newcommand{\Allamxn}{\Al_{\Lamxn}}%
\newcommand{\Allamxr}{\Al_{\Lamxr}}
\newcommand{\hatA}{\widehat{A}}
\newcommand{\hatB}{\widehat{B}}
\newcommand{\hatAt}{\widehat{A(t)}}
\newcommand{\hatBt}{\widehat{B(t)}}
\newcommand{\hatAg}{\widehat{A(g)}}
\newcommand{\hatBg}{\widehat{B(g)}}
\newcommand{\hatgA}{\widehat{\theta_g(A)}}
\newcommand{\hatgB}{\widehat{\theta_g(B)}}
\newcommand{\hatAlam}{\hatA_{\Lam}}
\newcommand{\hatBlam}{\hatB_{\Lam}}
\newcommand{\hatAglam}{\hatAg_{\Lam}}
\newcommand{\hatAlamiii}{\hatA_{\Lamiii}}
\newcommand{\hatBlamiii}{\hatB_{\Lamiii}}
\newcommand{\hatAlamii}{\hatA_{\Lamii}}
\newcommand{\hatBlamii}{\hatB_{\Lamii}}
\newcommand{\hatAlamn}{\hatA_{\Lamn}}
\newcommand{\hatBlamn}{\hatB_{\Lamn}}
\newcommand{\hatvpA}{{\hatA}^{\vp}}
\newcommand{\hatvpB}{{\hatB}^{\vp}}
\newcommand{\hatpsiA}{{\hatA}^{\psi}}
\newcommand{\hatpsiB}{{\hatB}^{\psi}}
\newcommand{\hatomeA}{{\hatA}^{\ome}}
\newcommand{\hatomeB}{{\hatB}^{\ome}}
\newcommand{\hatvpAg}{{\hatAg}^{\vp}}
\newcommand{\hatvpBg}{{\hatBg}^{\vp}}
\newcommand{\hatpsiAg}{{\hatAg}^{\psi}}
\newcommand{\hatpsiBg}{{\hatBg}^{\psi}}
\newcommand{\hatomeAg}{{\hatAg}^{\ome}}
\newcommand{\hatomeBg}{{\hatBg}^{\ome}}
\newcommand{\hatvpgA}{{\hatgA}^{\vp}}
\newcommand{\hatvpgB}{{\hatgB}^{\vp}}
\newcommand{\hatpsigA}{{\hatgA}^{\psi}}
\newcommand{\hatpsigB}{{\hatgB}^{\psi}}
\newcommand{\hatomegA}{{\hatgA}^{\ome}}
\newcommand{\hatomegB}{{\hatgB}^{\ome}}
\newcommand{\hatvpAinf}{\hatvpA_\infty}
\newcommand{\hatvpBinf}{\hatvpB_\infty}
\newcommand{\hatpsiAinf}{\hatpsiA_\infty}
\newcommand{\hatpsiBinf}{\hatpsiB_\infty}
\newcommand{\hatomeAinf}{\hatomeA_\infty}
\newcommand{\hatomeBinf}{\hatomeB_\infty}
\newcommand{\hatvpgAinf}{\hatvpgA_\infty}
\newcommand{\hatvpgBinf}{\hatvpgB_\infty}
\newcommand{\hatomegAinf}{\hatomegA_\infty}
\newcommand{\hatomegBinf}{\hatomegB_\infty}
\newcommand{\hatvpAginf}{\hatvpAg_\infty}
\newcommand{\hatvpBginf}{\hatvpBg_\infty}
\newcommand{\hatomeAginf}{\hatomeAg_\infty}
\newcommand{\hatomeBginf}{\hatomeBg_\infty}
\newcommand{\hatvpAtN}{{\hatAt}^{\vp}}
\newcommand{\hatvpBtN}{{\hatBt}^{\vp}}
\newcommand{\hatpsiAtN}{{\hatAt}^{\psi}}
\newcommand{\hatpsiBtN}{{\hatBt}^{\psi}}
\newcommand{\hatomeAtN}{{\hatAt}^{\ome}}
\newcommand{\hatomeBtN}{{\hatBt}^{\ome}}
\newcommand{\hatvpAtNinf}{\hatvpAtN_\infty}
\newcommand{\hatvpBtNinf}{\hatvpBtN_\infty}
\newcommand{\hatpsiAtNinf}{\hatpsiAtN_\infty}
\newcommand{\hatpsiBtNinf}{\hatpsiBtN_\infty}
\newcommand{\hatomeAtNinf}{\hatomeAtN_\infty}
\newcommand{\hatomeBtNinf}{\hatomeBtN_\infty}
\newcommand{\hatvpAinfKt}{\hatvpAinf(t)}
\newcommand{\hatvpBinfKt}{\hatvpBinf(t)}
\newcommand{\hatpsiAinfKt}{\hatpsiAinf(t)}
\newcommand{\hatpsiBinfKt}{\hatpsiAinf(t)}
\newcommand{\hatomeAinfKt}{\hatomeAinf(t)}
\newcommand{\hatomeBinfKt}{\hatomeBinf(t)}
\newcommand{\hatvpAt}{\hatvpA(t)}
\newcommand{\Lamxn}{\Lam_x(n)}
\newcommand{\Lamxr}{\Lam_x(r)}
\newcommand{\Lamxm}{\Lam_x(m)}
\newcommand{\LamtoG}{\Lam \uparrow \Gamma}%
\newcommand{\bLamtoG}{\bLam \uparrow \Gamma}%
\newcommand{\LamtoD}{\Lam \uparrow \Delta}%
\newcommand{\bLamtoD}{\bLam \uparrow \Delta}%
\newcommand{\LamitoG}{\Lami \uparrow \Gamma}%
\newcommand{\LamiitoG}{\Lamii \uparrow \Gamma}%
\newcommand{\LamiiitoG}{\Lamiii \uparrow \Gamma}%
\newcommand{\sequeZLamtoG}{\Lamn \uparrow \Zmu}%
\newcommand{\sequeLamtoG}{\Lamn \uparrow \Gamma}%
\newcommand{\sequeLamtoD}{\tilLamn \uparrow \Delta}%
\newcommand{\limG}{\LamtoG;\; \Lam \Subset  \Gamma}%
\newcommand{\limD}{\LamtoD;\; \Lam \Subset  \Delta}%
\newcommand{\sublimG}{\bLamtoG;\; \bLam \Subset  \Gamma}%
\newcommand{\netLamtoG}{\{\LamtoG;\; \Lam \Subset  \Gamma\}}%
\newcommand{\netLamG}{\{\LamtoG\}}%
\newcommand{\netLamtoD}{\{\LamtoD;\; \Lam \Subset  \Delta\}}%
\newcommand{\netLamD}{\{\LamtoD\}}%
\newcommand{\subnetLamtoG}{\{\bLamtoD;\; \bLam \Subset  \Gamma\}}%
\newcommand{\subnetLamtoD}{\{\bLamtoD;\; \bLam \Subset  \Delta\}}%
\newcommand{\Delnet}{\{\LamtoD;\; \Lam \Subset  \Delta\}}%
\newcommand{\mLam}{m_{\Lam}}%
\newcommand{\mLamA}{\mLam(A)}%
\newcommand{\mLamB}{\mLam(B)}%
\newcommand{\mLamAset}{\mLam(\{A_x; \; x\in \Gamma\}
)}%
\newcommand{\mLamBset}{\mLam(\{B_x; \; x\in \Gamma\}
)}%
\newcommand{\mLamaltA}{\mLam(\alpha_t(A))}%
\newcommand{\omemLam}{{m}^{\ome}_{\Lam}}%
\newcommand{\omemLamA}{\omemLam(A)}%
\newcommand{\omemLamAg}{\omemLam(\theta_g(A))}%
\newcommand{\omemLamB}{\omemLam(B)}%
\newcommand{\omembLam}{{m}^{\ome}_{\bLam}}%
\newcommand{\omembLamA}{\omembLam(A)}%
\newcommand{\omembLamAg}{\omembLam(\theta_g(A))}%
\newcommand{\omembLamB}{\omemLam(B)}%
%
\newcommand{\vpmLam}{m^{\vp}_{\Lam}}%
\newcommand{\vpmLamA}{\vpmLam(A)}%
\newcommand{\vpmLamB}{\vpmLam(B)}%
\newcommand{\vpmBacc}{\hatvpB_{\infty \, \netLamtoG}}
\newcommand{\rhoGinf}{\varrho^{\beta,\lim\rm{Gibbs}}}
\newcommand{\vacinf}{\rho^{\infty, \rm{Ground}}}
\newcommand{\rhoGlam}{\rho^{\beta,\rm{Gibbs}}_{\Lam}}
\newcommand{\rhoGlami}{\rho^{\beta,\rm{Gibbs}}_{\Lami}}
\newcommand{\varrhoGlam}{\varrho^{\beta,\rm{Gibbs}}_{\Lam}}
\newcommand{\vaclam}{\rho^{\infty, \rm{Ground}}_{\Lam}}
\newcommand{\corvpAB}{f^{\vp}_{A, B}}
\newcommand{\corpsiAB}{f^{\psi}_{A, B}}
\newcommand{\coromeAB}{f^{\ome}_{A, B}}
\newcommand{\corvpAA}{f^{\vp}_{A, A}}
\newcommand{\corpsiAA}{f^{\psi}_{A, A}}
\newcommand{\coromeAA}{f^{\ome}_{A, A}}
\newcommand{\grvpAB}{g^{\vp}_{A, B_r}}
\newcommand{\glvpAB}{g^{\vp}_{A_l, B}}
\newcommand{\glrvpAB}{g^{\vp}_{A_l, B_r}}
\newcommand{\cortilpsiAB}{f^{\tilde{\psi}}_{A, B}}
\newcommand{\cortilpsiAA}{f^{\tilde{\psi}}_{A, A}}

\newcommand{\corweakAB}{f^{\rhoGinf}_{A, B}}
\newcommand{\corweakABzero}{f^{\rhoGinf}_{A_0, B_0}}
\newcommand{\corometilAB}{f^{\ometil}_{A, B}}
\newcommand{\corlocGibbsAB}{f^{\beta, \rm{WOK}}_{\Lam;\, \hatAlam, \hatBlam}}%
\newcommand{\corlocvacAB}{f^{\infty, \rm{WOK}}_{\Lam;\, \hatAlam, \hatBlam}}
\newcommand{\altLam}{\alpha_{\Lam, \, t}}
\newcommand{\altLamii}{\alpha_{\Lamii, \, t}}
\newcommand{\altLamn}{\alpha_{\Lamn,\, t}}
\newcommand{\altLamm}{\alpha_{\Lamm,\, t}}
\newcommand{\altI}{\alpha_{\I,\, t}}
\newcommand{\altJ}{\alpha_{\J,\, t}}
\newcommand{\altIx}{\alpha_{\Ix,\, t}}
\newcommand{\hatH}{\hat{H}}
\newcommand{\hatU}{\hat{U}}
\newcommand{\hath}{\hat{h}}
\newcommand{\hx}{\hath_x}
\newcommand{\hzero}{\hath_{0}}
\newcommand{\Hlam}{\hatH_{\Lam}}
\newcommand{\Hlamfree}{\hatH_{\Lam}^{\rm{free}}}
\newcommand{\Hlamn}{\hatH_{\Lamn}}
\newcommand{\Hlamm}{\hatH_{\Lamm}}
\newcommand{\Home}{{H}_{\ome}}
\newcommand{\hatHvp}{\hat{H}_{\vp}}
\newcommand{\parext}{\partial_{{\rm{ext}}}}%
\newcommand{\parinside}{\partial_{{\rm{inside}}}}%
\newcommand{\lplus}{\Lam\cup \parext \Lam}%
\newcommand{\Allplus}{\Al_{\Lam\cup \parext \Lam}}%
\newcommand{\Lamkiku}{\breve{\Lam}}
\newcommand{\Lamhanpa}{\parext \Lamkiku}
\newcommand{\Lamkikueps}{\breve{\Lam}_\varepsilon}
\newcommand{\Lamhanpaeps}{\parext \Lamkikueps}

{{\textbf{Key Words}}} 
Genuine quantum time crystals; No-go theorem; KMS condition. \\
{{{Mathematics Subject Classification 2000: 82B03, 82B10}}}  
\section{Introduction}
\label{sec:INTRO}
 The  manifestation of  self-organized temporal periodic structures 
  of quantum  states was  proposed 
 by Wilczek  \cite{WIL}. 
The   crystal structures in the time direction   for \emph{equilibrium} 
 states   are   called 
\emph{genuine} quantum time crystals  \cite{KKlong}
to   distinguish   them from  non-equilibrium   
quantum  time crystals   which have been observed  
 experimentally 
  \cite{NATURE17MIK}  \cite{NATURE17MON}  \cite{NATURE21}.

  The existence of genuine  quantum time crystals was questioned
 soon after  the publication of  \cite{WIL}.
The work     \cite{BRUNOc, BRUNOnogo} by Bruno 
 verified   impossibility of spontaneously rotating time crystals
 for  the   original  quantum-time-crystal
 model  \cite{WIL}.
 The work  \cite{WO} by  Watanabe-Oshikawa
 provided    a general statement  of 
 absence of genuine quantum time crystals. 
In  \cite{WO} and  its  extension     
   \cite{wok}  by  Watanabe-Oshikawa-Koma (WOK), 
   non-trivial behaviors of   certain 
   temporal correlation functions 
  are identified  with  quantum time crystals.
  This criterion  
 of  quantum time crystals proposed by  
   \cite{WO, wok}  has been widely 
  used  in the   research   
  of  quantum time crystals, see    \cite{KHE} \cite{SACHA}.

In  this review, we   provide  
 other   no-go statements 
 of genuine quantum time crystals. We shall 
  compare them  with  the above mentioned 
 previous  works  \cite{WO, wok}. 
There are   several  reasons to do so. 
First, the  time-invariance  
for equilibrium states  (Proposition \ref{prop:INV}) is a  
 basic  fact of $\cstar$-algebraic quantum statistical mechanics
  established  already in  the  1970s 
 \cite{BR}. It 
  immediately    implies 
 the impossibility of genuine quantum time crystals  (Theorem \ref{thm:MAIN}). 
 That's the end of the story of \emph{genuine}  quantum time crystals.  
 Unfortunately, however, this   simple consequence of   
  \cite{BR}  has  been ignored   
 in  previous research works   on  quantum time crystals   
  \cite{HASACHA2} \cite{KHE} \cite{SACHA}.
Second, 
   the absence of temporal long-range order 
 (Corollary \ref{coro:NoLRO}) in the $\cstar$-algebraic formulation  
and  the  no-go  statement of \cite{WO, wok}
    are  different 
in   their   formulations and  mathematical  derivations, 
although they   look similar.
Third,  the 
 $\cstar$-algebraic   no-go statements
of genuine quantum time crystals
 are  more  thorough in terms of mathematical rigor and  
have much  wider generality than  \cite{WO, wok}. 
To establish clear   argument,   
 we shall   specify   exact meanings of  the following key   notions:

\begin{enumerate}
 \item Quantum systems
 \item  Equilibrium  
\item  Quantum time evolutions 
\item  Spontaneous symmetry  breakdown (SSB)
\end{enumerate}

For  the  first    in the above list, 
we formulate   our   quantum systems   by quasi-local  
   $\cstar$-systems. 
For the second, 
we define equilibrium  by  the   Kubo-Martin-Schwinger (KMS) condition. 
For the third, we formulate  
  Heisenberg  quantum  time evolutions 
   by   $\cstar$-dynamics.
For the fourth, the notion of SSB 
  is   defined by   the multiplicity  of KMS states (or ground states)
with respect to  the group action of a given  symmetry. 
 Namely, we use the   standard  $\cstar$-algebraic formalism
  \cite{BR}   which   will be  given  in Section~\ref{sec:SETUP}. 
 
In  the general $\cstar$-algebraic formulation as above, 
  non-existence  of time-translation symmetry 
 breakdown is almost obvious.
 As a consequence, the notion of 
genuine quantum time crystals  is  negated  with no effort.  
Even stronger, 
   any   non-trivial  order in the time direction,    
such as periodic,  quasi-periodic,  droplet, and  chaotic  order 
  is forbidden for  equilibrium states. 
 The precise  statement is 
  given in  Theorem \ref{thm:MAIN}  in  Section~\ref{sec:NON}. 
 
The   no-go statement   of genuine time crystals (Theorem \ref{thm:MAIN})
 is actually  not our finding;  
 it is   a   direct consequence of some basic  facts of    
  $\cstar$-algebraic quantum statistical mechanics \cite{BR}. 
As we will   see in this review, 
 the $\cstar$-algebraic formulation 
  has  several 
 advantages 
 over the   so called   
 ``the box-procedure method''   
 or   ``Gibbs Ansatz'' 
 widely  used in physics   to investigate 
 dynamical properties    
    of   equilibrium states (such as  genuine quantum time crystals). 
In terms of mathematics,  
 there are   several general  
    results  of the   KMS condition  \cite{BR}. 
 In terms of physics,  
  the  $\cstar$-algebra 
 gives  a  natural  formulation    of 
 quantum time evolutions which are generically non-local.
 We  refer  to    \cite{KAS78} by Kastler 
  for  some  conceptual points  of   $\cstar$-algebraic 
 quantum statistical mechanics.

This review is written 
  in a self-contained manner 
 so that the readers   can  understand  our statements  
 without referring to  the $\cstar$-algebraic literature.
We stress  that our  statements 
 will not  be   guaranteed,  unless   
 their    assumptions are completely    satisfied.
In particular, 
 the original   quantum crystal model \cite{WIL}    
 made by  many-particles  
on an Aharonov-Bohm  ring is beyond  the scope of 
our   $\cstar$-algebraic formulation.
We will   not  address 
   {\emph{non-equilibrium}} quantum time crystals, which 
  do  not conflict with   fundamental laws  of physics.

\section{Mathematical formulation}
\label{sec:SETUP}
In this section, we provide    our mathematical formulation  
based on  \cite{BR}.
We   make use of  
 the  $\cstar$-algebraic  definitions of equilibrium states, 
spontaneous symmetry breakdown, 
 and long-range orders.  
In addition to the basic reference \cite{BR}, 
 we may  refer to   \cite{SEW2014} \cite{SIM}.

\subsection{Quantum systems}
\label{subsec:QUNANTUMSYETEM}

Let  $\Gamma$  denote an infinite space  such as 
$\R^{\mu}$ and  $\Z^{\mu}$ ($\mu\in \NN$). 
$\Gamma$ has  a natural   additive  
group structure: $\xi_{x}(y):=y+x\in \Gamma$ for $x,y \in \Gamma$.
Let $\Fin$ be a  set of all subsets of $\Gamma$. 
If $\Lam\in\Fin$ has  finite volume $|\Lam|<\infty$,
 then we denote    `$\Lam\Subset \Gamma$'.
Let  $\Finf$ denote the  set of all 
 finite subsets 
(or the set of  sufficiently many finite subsets 
of  certain shapes)
 of $\Gamma$.
Let $\Al$ denote  a quasi-local $\cstar$-system on $\Gamma$
 describing the  infinite quantum system under consideration.
 The total system $\Al$ includes 
 a family of its subsystems 
$\{\All; \  \Lam\in \Fin\}$  indexed by $\Fin$. 
 The local algebra  $\Alloc:=\bigcup_{\Lam \in\Finf}\All$ 
is a norm  dense subalgebra of $\Al$.
 For any two disjoint subsets   $\Lam$ and $\Lamp$  the 
following  commutation relations are satisfied:
\begin{equation*}
  [A,\; B]\equiv AB-BA=0  \quad \forall
 A\in\All, \ \forall B \in \Allp.
\end{equation*}
The above condition  is called the local commutativity. 
Let $\{\tau_x\in\AUT,\;x\in \Gamma\}$ denote 
the group of  space-translation automorphisms  on $\Al$. 
The identity  $\tau_x (\All) =\Allx$ holds 
 for  any $\Lam\in\Fin$ and $x\in \Gamma$.
 
A state on $\Al$
is  a normalized positive linear 
 functional on $\Al$.
We denote the  set of all states on $\Al$ by 
$\State$. It is  an affine space with the affine combination of  states.
For each  $\ome\in \State$ 
 the triplet  $\bigl(\Hilome,\; \piome,\; \Omeome  \bigr)$
  denotes  the  Gelfand-Naimark-Segal (GNS)
 representation  associated with $\ome$.
 The GNS representation generates the von Neumann algebra  
$\vno:=\piome(\Al)^{\prime\prime}$ on the GNS Hilbert space $\Hilome$. 
The commutant of $\vno$  is given by  
$\vnomepri:=\bigl\{X \in \Bl(\Hilome); \  [X, \ Y]=0 \;
 \forall Y \in \vno \bigr\}$, and 
 the  center of $\vno$ is given by $\ZZ_\ome:=\vno \cap \vnomepri$.
The center $\ZZ_\ome$ contains  all macroscopic 
observables  with respect to  $\ome$. 
 Thereby  the center $\ZZ_\ome$ determines 
 the    macroscopic (thermodynamical) 
information about $\ome$. 
In general,  any  order parameter that  distinguishes different phases 
 has its   corresponding element  in   the center.
 For example,  the energy density,  
 the magnetization  per unit volume 
    for any  space translation invariant state, and  
  the  staggered magnetization per unit cell
  for any  space-periodic  state  belong  to the   center.

A state $\ome\in \State$
is called   a factor  state  
 if its center is trivial  $\ZZ_\ome=\CC \rm{I}$,
 where $\rm{I}$ is the identity operator on $\Hilome$. 
The set of all factor states on $\Al$ is denoted by $\FACState$.
Any   $\ome\in \FACState$ is known to satisfy   
the uniform cluster property with respect $\{\tau_x\in\AUT,\;x\in \Gamma\}$.  
 Hence   factor states  are    identified  with    pure phases. 
Each  $\ome\in \State$ has its  factorial (central) decomposition:
\begin{equation}
\label{eq:factor-dec}
\ome = \int  d\mu(\ome_{\lam})\ome_{\lam},\quad 
\ome_{\lam}\in \FACState, 
\end{equation}
where  $\mu$ denotes the  unique  probability  measure on 
 $\FACState$   determined  by $\ome$.
Note that  
 $\ome\in \State$ is not necessarily translation invariant.

\subsection{Quantum time evolution}
Assume that 
our  stationary Heisenberg-type  quantum  time evolution  is  given   by
a $\cstar$-dynamics, i.e.  
  a one-parameter group of automorphisms $\{\alpha_t\in\AUT,\;t\in \R\}$ 
 on  the quasi-local  $\cstar$-algebra  $\Al$.  
  We need   continuity for 
 $\{\alpha_t\in\AUT,\;t\in \R\}$ with respect to $t\in \R$.
We assume the strong continuity:  
\begin{equation}
\label{eq:strong}
 \lim_{t\to 0} \alpha_t (A)=A \ {\text{in norm for each fixed}} \ A\in \Al. 
\end{equation}
It is known that any  short-range  quantum spin lattice model 
 generates a strongly continuous time evolution 
  $\{\alpha_t\in\AUT,\;t\in \R\}$, see the pioneer work 
  \cite{ROB68}  and \cite{BR}.   
For  continuous  quantum systems,  we may  assume    
 $\sigma$-weakly  continuity for quantum   time evolutions
 in terms of the   GNS representation  of  sufficiently many
  states (including all equilibrium states) on $\Al$. 
In mathematics, these  quantum  systems 
are  called  $\wstar$-dynamical systems, see \cite{BR} \cite{D-SEW1970}. 

\subsection{Equilibrium states}
\label{sub:EQU}
\subsubsection{The KMS condition}
\label{subsubsec:KMS}
We define  equilibrium states  on $\Al$ 
  by the Kubo-Martin-Schwinger (KMS) condition,
 which  names after   Kubo \cite{KUBO},  Martin and Schwinger \cite{MS}.
We  capture  its mathematical
 formulation  due to
Haag-Hugenholz-Winnink \cite{HHW} in the following.

Let  $\beta \in \R_{+}\equiv \{a\in \R;\; a\ge 0\}$ denote an inverse temperature.    
Let  $\Dbeta:=\Bigl\{z \in \CC;\  0 \le {\Imag}z \le \beta    \Bigr\}$
   and  $\Dbetao:=\Bigl\{z \in \CC;\  0 < {\Imag}z < \beta
   \Bigr\}$. 
A state $\vp$ on $\Al$ is  called a 
 $\beta$-KMS state for  the quantum time evolution 
 $\{\alpha_t\in\AUT,\;t\in \R\}$
  if  there exists a complex function $F_{A,B}(z)$
 of $z \in \Dbeta$  for 
every $A,  B \in \Al$  such that 
 $F_{A,B}(z)$ is analytic in $\Dbetao$,
and the following relation holds; 
\begin{equation}
\label{eq:KMS-temporal}
F_{A,B}(t)=\vp \bigl(A \alpha_{t}(B) \bigr),\quad
F_{A,B}(t+i \beta)=\vp \bigl(\alpha_{t}(B)A \bigr), \quad
  \forall t\in \R.
\end{equation}
The set of $\beta$-KMS states 
 for  
 $\{\alpha_t\in\AUT,\;t\in \R\}$
 is denoted by  $\setKMSb$. 
Ground states are   defined  by the   
  KMS condition at  $\beta=\infty$.
 The set of ground  states  for  
 $\{\alpha_t\in\AUT,\;t\in \R\}$ 
 is denoted by  $\setGround$. 
 We denote  the set of all equilibrium states
 for   all  positive and infinite $\beta$  
 with respect to   the same quantum  time evolution
 $\{\alpha_t\in\AUT,\;t\in \R\}$
by $\setEQU$.

 The KMS condition  
is   a well-founded   definition of equilibrium in   quantum systems. 
Every  KMS state   satisfies    the minimum-free-energy condition. 
Vice versa, any  minimum-free-energy state given 
 by  the variational formula
 satisfies  the KMS condition, 
 see  \cite{ARKMSVAR}  \cite{ARAKISEWELL77}  \cite{RUE67}.
 (For  $\beta=\infty$,  the minimum energy condition   
for ground states is given  in \cite{BKR}.) 
 The  KMS condition implies  the 
   passivity formulated by Pusz-Woronowicz \cite{PWKMS}. 
 It  is   a kind  of    ``the second law of thermodynamics'': 
 No work can be obtained  from an isolated system 
 in  equilibrium by varying adiabatically  external parameters 
of the quantum time evolution  \cite{LEN}.

\subsubsection{Time invariance and  factorial  decomposition of equilibrium states}
\label{subsubsec:TIMEINV}
To discuss  temporal properties   of equilibrium states,  
 the KMS condition  has  several   advantages  over the 
  ``Gibbs Ansatz''    that is based on  local Gibbs ensembles on finite boxes,    since   the KMS condition   is  tightly related 
to  the quantum time evolution  by  definition.
As a notable  example,   the time invariance 
of \emph{any}   $\vp\in \setEQU$ follows from the  KMS-condition 
 \eqref{eq:KMS-temporal}, namely, 
\begin{align}
\label{eq:KMS-inv}
\vp\left(\alpha_t(A)\right)=\vp(A) \quad \text{for all } A\in \Al
 \quad (\forall t\in \R)
\end{align}
 holds 
regardless of whether 
 $\vp\in \setEQU$ is  a factor state (pure phase)  or not 
 (statistical mixture of multiple phases).
 
Next, 
  let us  address the structure of the set   of equilibrium states.  
For any  finite  $\beta\in \R_{+}$,  the affine  space  $\setKMSb$ 
 is  a Choquet simplex, and   the set of extremal points
$\EXsetKMSb$ in
 $\setKMSb$    coincides with  $\setKMSb \cap \FACState$.
Hence each  $\vp\in \setKMSb$ is uniquely 
 written as an affine sum of $\EXsetKMSb$: 

\begin{equation}
\label{eq:KMSfactor-dec}
\vp = \int  d\mu(\vp_{\lam})\vp_{\lam},\quad 
\vp_{\lam}\in \EXsetKMSb, 
\end{equation}
where  $\mu$ is  a unique  probability  measure 
  determined  by $\vp$. 
Any factor  state $\vp_{\lam}$ 
appearing  in the above factorial  decomposition \eqref{eq:KMSfactor-dec} 
 is a KMS state,  which  is obviously time-invariant. 
The above general structure of KMS states  
 follows from  the fact that 
  the center ${\ZZ}_\vp$ of any $\vp\in \setKMSb$
is pointwise  fixed   under the   time evolution \cite{ARmulti} \cite{BR}.

For $\beta=\infty$, a similar  statement  holds  as follows. 
  For any  $\vp\in \setGround$, consider      
 its state-decomposition 
$\vp = \int  d\mu(\vp_{\lam})\vp_{\lam}$, 
$\vp_{\lam}\in \State$. Then  each 
 $\vp_{\lam}$ belongs to $\setGround$ by the face property of $\setGround$.   
 Thus,  it  is obviously  invariant under the time evolution.
 The above general  structure of  ground  states   
 follows from  the fact that 
the commutant $\vnvppri$ (and therefore ${\ZZ}_\vp$) 
of any $\vp\in \setGround$
is pointwise fixed   under the time evolution \cite{ARbook} \cite{BR}.

\begin{remark}
\label{rem:MULTI}
The KMS condition  allows   
 multiple equilibrium phases (non-factor KMS states).
 Thus its  properties (such as the time-invariance)  hold 
 irrespective of the existence of  SSB.
\end{remark}

\begin{remark}
\label{rem:KMS-mukankei}
The KMS condition  does not require  specific dynamical assumptions such as  
 ergodicity, mixing, thermalization and so on.
 The work \cite{KHE}  suggested  
  special   treatments depending on 
 thermalization and non-thermalization dynamics  
  for   the proof of   \cite{WO}.
At least, it is sure that     
 such  divided treatments  are unnecessary for our statements. 
\end{remark}

\begin{remark}
\label{rem:KMS}
It looks that 
 the Gibbs-Ansatz is a more 
 familiar description  of equilibrium states on quantum systems
  than the KMS condition.
For  quantum spin lattice systems,  
 there is a  Gibbssian formulation  
 due to  Araki-Ion \cite{ARAKIION74}.
The Araki-Ion Gibbsian condition 
 is reduced  to the DLR condition  
for classical interactions.  Although
  the DLR condition is   a natural mathematical formulation 
  of  Gibbs states on  infinite classical 
systems \cite{SIM}, 
the  Araki-Ion Gibbsian condition  
  seems not allow such an  intuitive interpretation; 
   it is  an involved formula given 
 in terms of  the Dyson's perturbation method.\footnote{Prof. Araki told that 
 the Araki-Ion Gibbsian condition was 
an intermediate technical condition by which  
    the KMS condition and 
 the variational  principle can be  connected.} 
We only  mention  some recent proposals 
 \cite{Aff-preprint}  \cite{JAK}  of how to define     
``Gibbs  states'' in  quantum  spin lattice systems.
 Furthermore, the KMS condition
 may have wider  equilibrium states 
 than those determined by the  Gibbs ansatz  \cite{SMED}.
Thus, we use  the KMS condition as  the   
 primary notion  of quantum 
 equilibrium states.
We will come back to this point 
 in Section~\ref{subsubsec:Equilibrium-inclusion}.  
\end{remark}

 \subsubsection{Symmetry and spontaneous  symmetry breakdown}
\label{subsubsec:SSB}

We shall  make a  detour in order  to  establish 
  ``the absence  of  spontaneous  breakdown of time-translation
  symmetry'' which is essentially equivalent to  
   the time invariance 
 of equilibrium states  given in   Section~\ref{subsubsec:TIMEINV}. 
 We prepare   some  notations  related to the  
spontaneous  symmetry breakdown (SSB) in  the $\cstar$-algebraic language.

Let $G$ be a  group with its unit element $e$ and 
 let $(G, \theta)$ denote  a  faithful representation of $G$ into $\AUT$. 
Namely,   
\begin{align*}
\theta_g\in \AUT, \quad  g \in G, \nonumber \\
\theta_e ={\rm{id}}\in \AUT, \nonumber \\
\theta_g \ne {\rm{id}}\in \AUT \  \text{for}\ \forall g\ne e \in G,\nonumber \\ \theta_{g_1}\circ \theta_{g_2}=\theta_{g_1 g_2},\quad  g_1, g_2 \in G.  
\end{align*}
The  action of $G$ upon  $\State$ is given  by 
${\theta_g}^{\ast}\ome :=\ome\circ \theta_g\in \State$ ($g\in G$) 
for each  $\ome\in \State$. 
If ${\theta_g}^{\ast}\ome =\ome\in \State$ for all $g\in G$, then 
 $\ome$ is  called a $G$-invariant state.
The set of all $G$-invariant  states  is denoted by $\GState$.

A state   $\ome\in \State$ that is  invariant under the time evolution 
 $\{\alpha_t\in\AUT,\;t\in \R\}$ is called a time-invariant 
 or stationary state.
The set of all time-invariant  states $\STATIONState$
 will be denoted  simply by $\STATSTATE$.
As we have seen in \eqref{eq:KMS-inv}
\begin{align}
\label{eq:STATION}
\setEQU\subset \STATSTATE.
\end{align}

A state  $\ome\in \State$ that is  invariant under
the space-translation group 
$\{\tau_x,\;x\in \Gamma\}$  is  called 
 a  translation-invariant state. 
Let  $\Delta$ be a   crystallographic subgroup of $\Gamma$, 
an infinite sub-lattice   of $\Gamma$ such that 
 the quotient group $\Delta/\Gamma$ is finite.
If  $\ome\in \State$ is invariant under 
 $\{\tau_x,\;x\in \Delta\}$,  
then it is called  a  spatially  periodic  state with respect to  $\Delta$. 
We denote  the sets  of all translation-invariant states
 and  all spatially periodic states
  by  $\TRASTATE$ and  $\perSTATE$, respectively. 
Together, translation invariant and  spatially
 periodic states are   called 
 homogeneous  states. 
We denote  the set of all  homogeneous  states on $\Al$ by $\HOMOSTATE$.

We define  spontaneous symmetry breakdown (SSB)   as follows.
If 
\begin{equation}
\label{eq:DYNSYM}
\alpha_t \circ \theta_g=  \theta_g\circ \alpha_t  \in \AUT
 \quad {\text{for all}} \ t\in \R,\ g\in G,
\end{equation}
then $(G, \theta)$  is called a dynamical symmetry group.
For  such  $(G, \theta)$, if there exists  a state 
  $\psi\in \setKMSb$  breaking the symmetry,  
 namely ${\theta_g}^{\ast} \psi\ne \psi$ for some $g\in G$,
then  the  dynamical symmetry group $(G, \theta)$ is  
  spontaneously broken.

\begin{remark}
\label{rem:BOGO}
The above definition of SSB  is based 
 on   the set of  equilibrium states $\setKMSb$ on the $\cstar$-system $\Al$.  
 There is another  
 definition of SSB  
  called  the Bogoliubov's method \cite{BOGO}.
The relationship between these  different definitions  of SSB
 for spin lattice systems 
has been elucidated in   III. 10 of \cite{SIM} and  Theorem 6.2.42 of 
\cite{BR}.  For   boson systems,  
 known  facts  are rather limited in comparison with  the 
 case of spin lattice 
 systems, see  \cite{WZ}.
\end{remark}

\section{Nonexistence of time-translation  symmetry breakdown}
\label{sec:NON}

\subsection{General no-go statements}
\label{subsec:GENERAL}

In  the preceding section,  
we  recalled the following basic fact 
 of  equilibrium states.
\begin{prop}
\label{prop:INV}
Any  equilibrium state on  $\cstar$-algebras 
  satisfying   the KMS condition 
is   invariant under the time evolution.
Each  macroscopic observable in the center of the von Neumann  algebra 
 generated by the GNS representation  
is fixed under the time evolution,  irrespective of whether 
such a  state 
 is a factorial state (pure phase) or a non-factorial state 
 (statistical mixture of different  phases).
\end{prop}

\begin{remark}
\label{rem:MOVENOT}
The  frozen property of  equilibrium  states
  stated  in Proposition \ref{prop:INV} 
 may be   expressed as:  ``Any macroscopic order parameter  moves not!"
 This is  the  subtitle of this review.
 Here  macroscopic order parameters 
 are identified with  elements in  the center
 which is a classical (commutative) algebra.
Those are  usually given by the thermodynamic limit of densities.
A similar statement  was established  in  the work  \cite{BRUNOmove}
   for  the  particular   model  \cite{WIL}.
In \cite{NAKA},  a  non-equilibrium   quantum  time crystal model 
  similar to 
 \cite{WIL} was  developed
 by using a non-conventional definition 
of SSB.  This is a
 meta-stable state  
which does  not conflict  with the above   no-go statement.
\end{remark}

The impossibility of 
spontaneous  breakdown for time-translation symmetry  immediately 
follows from Proposition \ref{prop:INV} as follows.

\begin{theorem}
\label{thm:MAIN}
Suppose that a  quantum time evolution is given 
 by  $\cstar$-dynamics,  and equilibrium states 
 are given  by  the KMS condition with respect to the  quantum time evolution.
Then there exists  no  spontaneous 
breakdown of  time-translation symmetry, and 
 therefore,   no  temporal  order exists in  equilibrium states.  
In particular,  periodic, 
 quasi-periodic, and  chaotic orders  in the time direction 
 are  forbidden in  equilibrium states.
\end{theorem}

In the above theorem,   the  one-parameter group of automorphisms 
 $\{\alpha_t\in\AUT,\;t\in \R\}$ 
 plays   two roles. First, it denotes  
 a Heisenberg  quantum time evolution  determining  
 the quantum model under consideration.  Second, it provides   
 a special  dynamical symmetry $(G, \theta)=(\R, \alpha)$ of  the model, 
 as the requirement \eqref{eq:DYNSYM}
 is satisfied by  the  following obvious commutative  relation:
\begin{equation*}
\alpha_t \circ \alpha_s=  \alpha_s \circ \alpha_t =\alpha_{t+s} \in \AUT
 \quad {\text{for all}} \ t\in \R,\ s\in \R.
\end{equation*}

\begin{remark}
\label{rem:trace}
Although it is unnecessary, 
  we shall take a look at the  $\beta=0$ case (the 
infinite temperature) 
   to compare  
 Theorem \ref{thm:MAIN}   with 
 \cite{wok} that provides 
  special treatment  for this case  in II. C ``Proof for $\beta=0$''. 
The KMS condition for $\beta=0$ yields  
 the identity  $\vp (A B)=\vp ( BA)$
 for all $A, B\in \Al$, so $\vp$ is  identical 
 to the tracial state $\tr$ on $\Al$.  
The tracial state is a factor state (if $\Al$ is 
 a simple algebra as usual in quantum statistical mechanics). 
Hence,   it  
satisfies   
the uniform cluster property with respect the 
 space-translations. 
 As the tracial state  is  invariant under any automorphism, 
it is automatically invariant  under 
$\{\alpha_t\in\AUT,\;t\in \R\}$. 
Hence the possibility of quantum time crystal is obviously negated. 
\end{remark}

\begin{remark} 
\label{rem:GRAND}
 KMS states for    correspond to 
canonical ensembles  determined by the inverse temperature $\beta\in \R_{+}$.
In a similar  manner, one can  consider  the 
KMS condition corresponding to   grand canonical ensembles 
 specified by  $\beta$
 and the chemical potential(s) $\mu$, 
  see \cite{HHW} and  Sec. 5.4.3 of \cite{BR}.
For grand 
 canonical KMS states as well, 
we can obtain    statements similar  to     
Proposition \ref{prop:INV} 
  Theorem \ref{thm:MAIN} and  Corollary \ref{coro:NoLRO}.  
If U(1)-symmetry  generated by the number operator 
breaks spontaneously such  as  the Bose-Einstein  condensation, then 
 the expectation value of the field operator which  is responsible 
 for the U(1)-symmetry  breakdown will oscillate  
periodically in time  as  noted in \cite{VOL} \cite{wok}. 
However,   the time invariance property  \eqref{eq:KMS-inv}
 holds   for all gauge invariant observables, i.e.  elements  in  $\Al$
 that are fixed  by the U(1)-transformation.
\end{remark}

\subsection{Absence of   temporal long-range order}
\label{subsec:LROcstar}
Proposition \ref{prop:INV} and Theorem \ref{thm:MAIN}   require 
 no  particular assumption on  the spatial structure.
In this subsection,  we specialize in the homogeneous case, where 
 the  states and  the quantum time evolution under consideration  
are both assumed to be   homogeneous (translation invariant or 
  spatially  periodic).

\subsubsection{Densities as macroscopic observables}
\label{subsubsec:densities}
 Each   local observable  gives  a macroscopic observable 
 on the GNS Hilbert space for any 
homogeneous (translation invariant or 
  spatially  periodic) state  by  its Ces\` aro sum  under 
  space translations.
Let us  give this   statement a precise  
  formula  below following \cite{HAAGBCS} \cite{BR}.

We  prepare  some  notations regarding 
infinite-volume  limits. 
Let  $\Gamma$  denote our  infinite space, 
 and  $\Delta$ denote  a  crystallographic  subgroup of $\Gamma$ as before.
Suppose that  a net $\{\Lam;\; \Lam \Subset \Gamma\}$ of finite subsets 
 of $\Gamma$  eventually includes any  $\I \Subset \Gamma$. 
Then,  we denote   
$\netLamtoG$ or  simply $\netLamG$.
Similarly,  if  a  net $\{\Lam;\; \Lam \Subset \Delta\}$
  eventually includes any  $\I \Subset \Delta$, then  we denote
$\netLamtoD$  or   simply $\netLamD$.

Take any $A\in \Al$. 
For  any  finite subset $\Lam\Subset \Gamma$ (or $\Lam\Subset \Delta$)
we take  the following  averaged sum:
\begin{equation}
\label{eq:mLamA}
\mLamA:= 
 \frac{1}{|\Lam|} \sum_{x\in\Lam} 
\tau_{x}(A)\in \Al.
\end{equation}
For $\ome\in \State$ 
let $\omemLamA:= \piome 
\left(\mLamA \right) \in \vnome$.
 For any  $\ome \in \TRASTATE$
 and  any $\netLamtoG$, 
the   net of uniformly bounded operators  
 $\bigl\{\omemLamA\in \vnome;\; \limG  \bigr\}$ 
has  at least one and more   accumulation point(s) 
 in the center $\ZZ_{\ome}$. Precisely,  
 there exists    a subnet  
 that converges to some element  
of  the center  in the weak-operator topology.
We denote  any such  accumulation point   by $\hatomeAinf$.
Heuristically, we write  
\begin{equation}
\label{eq:mAacc}
\hatomeAinf\equiv 
  \lim_{\sublimG}\omembLamA \in \ZZ_{\ome},
\end{equation}
where 
$\subnetLamtoG$ is some  subnet of  $\netLamtoG$.  
Any  such $\hatomeAinf$ belongs to the center
 $\ZZ_{\ome}$ and is called  a
macroscopic observable (thermodynamic density, or observable 
at infinity \cite{LAN-RUE}).
When  $\ome$ is a  non-factor state, 
 there can be  multiple accumulation points. 
 Similarly, we consider  $\ome \in \perSTATE$. 
 Then every  accumulation point of 
$\bigl\{\omemLamA\in \vnome; \; \limD \bigr\}$ 
in the weak-operator topology  belongs to $\ZZ_{\ome}$.
If  $\ome \in \HOMOSTATE$ is  a factor state, 
these macroscopic observables  
are  sharply given as the scalars $\ome(A){\rm{I}}$  with no dispersion 
 for   $A\in\Al$.

\subsubsection{Long-range order in the $\cstar$-algebraic formulation}
\label{subsubsec:LROCstar}
We give  a general   formulation of   
 long-range order (LRO) in the  $\cstar$-algebraic formulation.
Take  $\ome \in \HOMOSTATE$ and  $A, B\in \Al$.  
Denote their  densities by  $\hatomeAinf, \hatomeBinf\in \ZZ_{\ome}$
 as in \eqref{eq:mAacc}.
Consider  the following two-point correlation function 
with respect to $\ome$ 
\begin{equation}
\label{eq:coromeAB}
\coromeAB \equiv 
\left(\Omeome,\; \hatomeAinf
\hatomeBinf \Omeome\right),
\end{equation}
 where the  GNS representation of $\ome$ is used. 
If   $\coromeAB$ is non-zero, 
then $\ome$ is said to exhibit  LRO for  $A, B\in \Al$.
If   $\coromeAA$ is non-zero for some 
 $A\in \Al$, then $\ome$  is said to exhibit  LRO 
and  $A\in\Al$ is called  a local  order parameter.

Now we  consider the group action $(G, \theta)$. 
For  $A\in\Al$ and $g\in G$, we 
define
\begin{equation}
\label{eq:gmacro}
\hatomeAinf(g):=\hatomegAinf\in \ZZ_{\ome}, 
\end{equation}
 by substituting  
$\theta_g(A)$  for   $A$ in \eqref{eq:mAacc}. 
 We consider   the following  two-point correlation function 
with respect to  $\ome \in \HOMOSTATE$:
\begin{equation}
\label{eq:GcoromeAB}
\coromeAB(g)\equiv 
\left(\Omeome,\; \hatomeAinf(g)
\hatomeBinf \Omeome\right),\  g\in G.
\end{equation}
If $\coromeAB(g)$ is  a non-constant function of $g\in G$, 
then $\ome$ is said to exhibit   $G$-dependent  LRO for 
 $A, B\in\Al$.
If $\coromeAA(g)$ is  a non-constant function of $g\in G$
 for some  $A\in \Al$, 
then $\ome$ is said to exhibit   $G$-dependent  LRO, 
 and  $A\in\Al$ is called  
 a local order parameter  with respect to  $(G, \theta)$-symmetry.

Next,  we consider  the quantum time evolution.  
Let $\vp\in  \setEQU\cap \HOMOSTATE$, i.e. an arbitrary   homogeneous 
equilibrium state for $\{\alpha_t\in\AUT,\;t\in \R\}$.  
 Substituting $\vp$ for $\ome\in \HOMOSTATE$, 
$(\R, \alpha)$ for  $(G, \theta)$, and $t\in \R$ for   
 $g\in G$  of  $\hatomeAinf(g)$ defined above, 
we obtain a  macroscopic observable    
\begin{equation*}
\hatvpAinf(t)\in \ZZ_{\vp}. 
\end{equation*}
Let  $A, B\in \Al$. 
Consider   the following two-point temporal correlation function 
with respect to $\vp$:
\begin{equation}
\label{eq:corvpAB}
\corvpAB(t)\equiv 
\left(\Omevp,\; \hatvpAinfKt
\hatvpBinf \Omevp\right),\  t\in \R.
\end{equation}
By the pointwise invariance of the center $\ZZ_{\vp}$ under 
 the time evolution $\{\alpha_t\in\AUT,\;t\in \R\}$ as 
 stated  in Proposition \ref{prop:INV},   
\begin{equation}
\label{eq:time-KMSfix}
\hatvpAinfKt=\hatvpAinf, \quad \forall t\in \R.
\end{equation}
It implies  the  fixed temporal correlation function:
\begin{equation}
\label{eq:tempinv}
\corvpAB(t)=
\corvpAB(0)\quad \text{for all}\  t\in \R.
\end{equation}
Thus, the  absence of 
 non-trivial temporal  LRO is  proved.
\begin{corollary}
\label{coro:NoLRO}
Assume  the same  assumption of  Theorem \ref{thm:MAIN}. 
Assume further that 
the  time evolution is  homogeneous in space.
Then,  there exists   no non-trivial temporal LRO
 for any   homogeneous equilibrium state. 
\end{corollary}

\begin{remark} 
\label{rem:NONLOCAL}
We have defined 
$\hatomeAinf\in \ZZ_{\ome}$
 for all $A\in \Al$ in 
 \eqref{eq:mAacc} which  are not necessarily strictly 
 local. 
This is essential for 
 $\corvpAB(t)$  in \eqref{eq:corvpAB} to be well defined, 
 since generically the 
 time development $\alpha_t(A)$ of $A\in \Alloc$   
does not  stay   in $\Alloc$. 
\end{remark}

\begin{remark} 
\label{rem:CORO}
 Corollary \ref{coro:NoLRO}  has  the following obvious  generalizations.
Let  $\vp\in  \setEQU\cap \HOMOSTATE$
 and  $A, B\in \Al$ as in Corollary \ref{coro:NoLRO}.  
 Consider 
\begin{align*}
\left(\Omevp,\; \hatvpAinfKt
\pivp(B) \Omevp\right),\ t\in \R,\\ 
\left(\Omevp,\; \pivp\left(\alpha_t(A)\right) 
\hatvpBinf
\Omevp\right),\  t\in \R. 
\end{align*}
Then 
  from  \eqref{eq:time-KMSfix} and the time-invariance of $\vp$
it follows that these 
temporal two-point 
 correlation functions are  constant  with respect to $t\in\R$.
 On the other hand, generically, the  two-point function   
$\left(\pivp\left(
\alpha_t(A)\right) 
\pivp(B) \Omevp\right)$ 
is  not   constant in  $t\in\R$
 as noted  in \cite{WO}.
\end{remark}

\section{Comparison}
\label{sec:COMPARISON}
In this section,  we  compare our 
no-go statements of 
genuine quantum time crystals
  given in   Section \ref{sec:NON}
with  the previous works  by   
 Watanabe-Oshikawa-Koma which will be summarized 
 in the  first subsection.
 
\subsection{Summary of  the result of  Watanabe-Oshikawa-Koma}
\label{subsec:WOK}
We  recall   the formulation and the  main result   
 of   \cite{WO, wok}  
 adding some  minor   modifications for comparison purposes. 
 We    always explicitly write  the  $\Lam$-dependence of subsystems
 as $\All$, as we consider such  specification is crucial. 
 Furthermore,  any subsystem  $\All$  is embedded 
 into the  total system  $\Al$.  
We shall take  
 the cubic lattice $\Zmu$ as in \cite{wok}. 
Define the metric on $\Gamma\equiv \Zmu$ by 
$\Vert x-y \Vert:=\max_{ i \in \{1,2, \cdots, \mu\}}|x_i-y_i|$
 for $x=(x_i),\; y=(x_i)\in \Gamma$.
 Let  ${\rm{diam}}(\Lam)\equiv \{\sup{\Vert x-y\Vert ; \; x,y\in \Lam}\}$
 for $\Lam\subset \Gamma$.

Suppose  that the Hamiltonian on the total system  
 is formally  given by 
\begin{equation}
\label{eq:haH}
\hatH:=\sum_{x\in \Gamma} \hx,
\end{equation}
 where  each local Hamiltonian $\hx\in \Alloc$ is a finite-range  
self-adjoint operator with  its 
 support ${\rm{supp}}(\hx)$  centered at  site $x\in \Gamma$. 
We assume that
the range and the norm of $\{\hx; \; x\in \Gamma\}$
 are uniformly bounded over $x\in \Gamma$:
\begin{align}
\label{eq:Hxuniform}
{\rm{diam}}({\rm{supp}}(\hx))\le R_h,\quad  
 \Vert  \hx  \Vert\le  N_h,  
\end{align}
where $R_h$ and $N_h$
 are some  positive constants.

For each  $\Lam\Subset \Gamma$, let 
\begin{equation}
\label{eq:Hlam}
\Hlam:=\sum_{x\in \Lam} \hx \in \Allplus,
\end{equation}
where  $\parext \Lam$ denotes some outer surface  region 
surrounding but   not intersecting $\Lam$.
 By  \eqref{eq:Hxuniform} the ration $\frac{|\parext\Lam|}{|\Lam|}$ 
tends to $0$  as $\LamtoG$. 
One may   choose  the free-boundary local Hamiltonian  
$\Hlamfree:=\sum_{x\in \Lam \setminus \parinside} \hx \in \All$,
 where $\parinside$ is the  smallest subset within $\Lam$
 such that the above sum is in $\All$.
In the following, we  use   the above local Hamiltonian 
 \eqref{eq:Hlam} as in  \cite{wok}. 

For $\Lam\Subset \Gamma$, we define 
the local Heisenberg time evolution 
 by 
\begin{equation}
\label{eq:HeisenLam}
\altLam(A):=
e^{it  \Hlam} A e^{-i t \Hlam}
\ \text{for }\  A\in  \Al, \quad t\in \R.
\end{equation}
For $\Lam\Subset \Gamma$, we  
define the local Gibbs state $\rhoGlam\in \State$
 at inverse temperature $\beta$ by
   the same local Hamiltonian  $\Hlam$
as 
\begin{equation}
\label{eq:rhoGlam}
\rhoGlam(A):=
\frac{1}{\tr(e^{-\beta \Hlam})} \tr(e^{-\beta \Hlam}A)
 \quad  \text{for }\  A\in  \Al,
\end{equation}
 where $\tr$ denotes  the tracial state on $\Al$.  
Note that the  local Gibbs state 
$\rhoGlam$  defined  on the total system $\Al$ 
 is the  unique $\beta$-KMS state for   
 $\{\altLam\in\AUT,\;t\in \R\}$.
Similarly,  let  
$\vaclam\in  \State$ denote a  ground state for  
the  local Hamiltonian $\Hlam$, equivalently, 
 a ground state  for   
 $\{\altLam\in\AUT,\;t\in \R\}$ as defined in  
Section~\ref{subsubsec:KMS}.
Note that such ground state surely exists 
 $\vaclam\in  \State$ if 
 the  local Hamiltonian $\Hlam$ is a bounded element;
 this is always the case for quantum spin lattice systems.
It is often the case  that there  non-unique  
 ground states.  

Take a set of local operators      
 $\{A_x; \; x\in \Gamma\}$, where each  $A_x$ is a local operator  
with its  support  ${\rm{supp}}(A_x)$
   centered at   $x\in \Gamma$.
Assume  that  
the range and the norm for $\{A_x; \; x\in \Gamma\}$
 are uniformly bounded over $x\in \Gamma$:  
there are  positive constants $r$ and  $a$ such that 
 for all  $x\in \Gamma$  
\begin{align}
\label{eq:Axuniform}
{\rm{diam}}({\rm{supp}}(A_x))\le r,\quad  
 \Vert  A_x \Vert\le  a.  
\end{align}
Then for  $\Lam\Subset \Gamma$, we define  
\begin{align}
\label{eq:hatAlam}
\hatAlam \equiv \mLamAset:= 
 \frac{1}{|\Lam|} \sum_{x\in\Lam} 
A_x \in \Alloc.
\end{align}
 The  notation $\hatAlam$ above  corresponds  to  
 $\hatA$   in  \cite{wok}. 
If 
 $A_0$ is a local operator  
with its  support   centered at  the origin and 
$A_x=\tau_x(A_0)$ for all $x\in\Gamma$, 
 then such $\{A_x; \; x\in \Gamma\}$ is said to be  
  covariant in space-translations, and 
$\hatAlam$  is equal   to  $\mLamA$ of 
 \eqref{eq:mLamA} with  $A:=A_0\in \Al$.

Consider  $\{A_x; \; x\in \Gamma\}$
  and $\{B_x; \; x\in \Gamma\}$, both 
 satisfying  \eqref{eq:Axuniform}, and then take  their    
 $\hatAlam$ and $\hatBlam$ as in \eqref{eq:hatAlam} 
for each $\Lam \Subset \Gamma$. 
For each $\Lam \Subset \Gamma$,  define   the following temporal two-point 
 correlation function for the local Gibbs state \eqref{eq:rhoGlam}
 under the local Heisenberg time evolution \eqref{eq:HeisenLam}:
\begin{equation}
\label{eq:WOKtemp}
\corlocGibbsAB(t)\equiv 
\rhoGlam\left(\altLam\bigl( \hatAlam \bigr) \hatBlam  \right), \quad t\in \R. 
\end{equation}
Similarly, for the case of $\beta=\infty$,  let
\begin{equation}
\label{eq:WOKground}
\corlocvacAB(t)\equiv 
\vaclam \left(\altLam\bigl( \hatAlam \bigr) \hatBlam  \right), \quad t\in \R. 
\end{equation}

The above  functions will be  called  WOK temporal  correlation functions. 
The main statement  of  \cite{WO, wok} is  as follows.
For each fixed  $t\in \R$, 
\begin{equation}
\label{eq:WOK-temp-nai}
\lim_{\LamtoG} \left|\corlocGibbsAB(t)-\corlocGibbsAB(0)\right|=0\quad \text{for any } \ \beta\in \R_{+}, 
\end{equation}
and 
\begin{equation}
\label{eq:WOK-gr-nai}
\lim_{\LamtoG} \left|\corlocvacAB(t)-\corlocvacAB(0)\right|=0
\end{equation}
 are satisfied. 
By the  triviality  of  the   
  (Griffiths-type) LRO with respect to  $t\in \R$ 
as in   \eqref{eq:WOK-temp-nai} \eqref{eq:WOK-gr-nai}, 
Watanabe-Oshikawa-Koma concluded  
 ``absence of quantum time crystals for equilibrium states''.

\begin{remark}
\label{rem:limitWOK}
Equation~\eqref{eq:WOK-temp-nai}  
does
 not assert  the identity  $\lim_{\LamtoG} \corlocGibbsAB(t)
=\lim_{\LamtoG}\corlocGibbsAB(0)$.
The existence of this limit  is not known.  
\end{remark}

\subsection{On different formulations  of LRO}
\label{subsec:differentLRO}
 There are variant   formulations of 
LRO (long-range order). 
The well-known definition  
based on the box procedure-method is due to  Griffiths  \cite{GRIF}. 
Precisely,  it  is   formulated  by 
 a net  of  local Gibbs states (under some boundary condition)
 together with  averaged local observables.
 Another  definition of   LRO  
 is  defined in terms of   states on   a quasi-local $\cstar$-algebra 
and macroscopic observables (which belong to the von Neumann 
 algebra not in the given $\cstar$-algebra), we refer to  \cite{SEW2014}.
Let us call the former    
 the Griffiths-type  LRO,    
and the latter   
the  $\cstar$-algebraic   LRO. 
 The Griffiths-type  LRO
  has produced  remarkable  results  on   
 several  statistical-physics models \cite{GRIF} \cite{DLS}, 
 whereas the  $\cstar$-algebraic   LRO
  is  a   mathematical formula
 which is useful  for  general discussion   
but  not for  practical analysis of concrete models.   
Corollary \ref{coro:NoLRO}   
is  based on the $\cstar$-algebraic   LRO, 
 whereas  
 the  work  \cite{WO, wok} relies  
 on the  Griffiths-type LRO as we have seen  in Section~\ref{subsec:WOK}.   
 Before going into the in-depth discussion, let us  recall   
general information on 
 these  two   different  LROs in terms of  SSB:

\begin{itemize}   
\item  
The existence of non-trivial  $\cstar$-algebraic LRO 
is  \emph{equivalent}
to the existence of  multiple phases. 
Here,  the  states      
 are  assumed to be homogeneous states on the quasi-local $\cstar$-algebra 
$\Al$, but  not necessarily equilibrium states. 
See $\S$5.2 of \cite{SEW2014} for  this equivalence relation. 
 
\item 
Non-trivial  Griffiths-type LRO   appeared in a  (classical or quantum) 
 spin lattice model implies 
a corresponding spontaneous  symmetry breakdown.
 For the   precise statement,   see   
 Theorem 1.3 of \cite{DLS}, \cite{KT-LONGRANGEORDER}, \cite{KT-Finite}, 
 and   $\S$5.5 of \cite{SEW2014},  Sec.III.10 of \cite{SIM}.
 On the other hand, the converse implication 
 is not known in general (even for classical lattice models). 
From the state of the art of mathematical rigorous statistical mechanics, 
  there are few  non-trivial cases for which 
 the converse implication is justified. 
\end{itemize}   

Watanabe-Oshikawa-Koma 
defined 
 a Griffiths-type LRO 
 with the time parameter by  
$\lim_{\LamtoG} \corlocGibbsAB(t)$
 ($t\in \R$).
 And they   postulated  that 
 the  
 periodicity of this quantity  in  time 
identifies with emergent  quantum time crystal.
However, 
 in view of  the  general status of Griffiths-type LRO mentioned above, 
  the absence of 
 time translation symmetry breakdown for equilibrium states  
 cannot be concluded solely by   the triviality of 
 the Griffiths-type LRO  \eqref{eq:WOK-temp-nai} \eqref{eq:WOK-gr-nai}.
As far as we understand,   
 the  essential idea of    \cite{WO, wok}
is owing   to  the   method  of  
finding  SSB  in  quantum spin lattice models  given   in   \cite{KT-Finite}.
However, it is  not certain  whether 
 non-detection of SSB  by this   specific method   
   yields   a  {\emph{complete}}  proof of the   absence of SSB.
On the other hand, the KMS condition  
 completely excludes temporal SSB.

\subsection{On limit procedure}
\label{subsec:Limit}
 Our  $\cstar$-algebraic formulation    and  the  
works \cite{WO, wok}
are very different  in  the treatment of the infinite-volume-limit.
We now  introduce some   notaions 
  concerning   the  infinite-volume limit
 and  recall some  related  facts.

  Let   $\rhoGinf$  denote  
 an arbitrary accumulation point 
of the net of 
 local Gibbs states $\{\rhoGlam; \;  \Lam\Subset \Gamma\}$
 \eqref{eq:rhoGlam}. 
 Heuristically, we write   
\begin{align}
\label{eq:INFINITE-GIBBS}
\rhoGinf(A)=
\lim_{\LamtoG}
\rhoGlam(A), \quad A\in \Al.
\end{align}
Any  accumulation point $\rhoGinf$ is  called a limiting Gibbs state.
Such $\rhoGinf$ is not necessarily unique.
Let $\setACCUM$ denote the set of all such $\rhoGinf$. 
The quantum time evolution  
  $\{\alpha_t\in\AUT,\;t\in \R\}$  is 
 called   approximately inner if 
\begin{equation}
\label{eq:APP}
\alpha_t(A)=\lim_{\LamtoG}
\altLam(A) 
 \quad  \text{for each}\  A\in \Al
  \ \text{and }\  t\in \R,
\end{equation}
where 
$\altLam$  denites 
 the local Heisenberg time evolution \eqref{eq:HeisenLam}, 
 and 
the convergence  is with respect to the  norm 
(or $\sigma$-weak topology introduced by  the  GNS representation of a chosen 
 state).
The  existence of 
 at least one and more  limiting    Gibbs states $\rhoGinf$ 
   as in \eqref{eq:INFINITE-GIBBS}
and the existence   of a unique strongly continuous 
 approximately inner $\{\alpha_t\in\AUT,\;t\in \R\}$ 
 as in  \eqref{eq:APP}
   have   been  verified 
 for any   short-range  quantum spin lattice model, 
see \cite{ROB68} \cite{POWSAK}, and  Theorem 6.2.4 of \cite {BR}. 
Also, it has been  known    that under the same assumption 
  every  limiting Gibbs state  $\rhoGinf$
  satisfies the   KMS condition   with respect to
  $\{\altLam\in\AUT,\;t\in \R\}$ for $\beta$.
Thus,  the following inclusion holds: 
\begin{equation}
\label{eq:ACCUsmall}
\setACCUM \subset \setKMSb
\end{equation}

\subsubsection{How to define  LRO in the infinite volume limit?}
\label{subsubsec:HowtoLRO}
The WOK temporal  correlation function
$\corlocGibbsAB(t)$
defined  in \eqref{eq:WOKtemp}
  has  the  same $\Lam$-dependence on   
 the local Gibbs  state $\rhoGlam$, the 
cut-off time-translation symmetry
  $\altLam$,  and  the   local order parameters 
$\hatAlam$ and  $\hatBlam$.
These three $\Lam$s  are taken  to  infinity  
 at the same time. As  noted in   Remark \ref{rem:limitWOK}, however,  
the existence of  $\lim_{\LamtoG}\corlocGibbsAB(t)$
  has not been  verified. On the other hand, 
if  the three  limits are taken in the  specified order
as follows,   then    a   $\cstar$-algebraic LRO (whose 
 existence is surely verified for  short-range 
  quantum spin-lattice models) 
 will appear: 
\begin{equation}
\label{eq:tripleLIMITs}
\lim_{\LamiiitoG}
\lim_{\LamiitoG}
\lim_{\LamitoG}
\rhoGlami\left(\altLamii\bigl( \hatAlamiii \bigr) \hatBlamiii  \right)
=\corweakABzero(t), 
\end{equation}
where the right-hand side  is given by the formula 
 $\corvpAB(t)$
\eqref{eq:corvpAB} 
 in  Section~\ref{subsubsec:LROCstar}
   with  the approximately inner time evolution $\{\alpha_t\in\AUT,\;t\in \R\}$ for the chosen $\cstar$-dynamics, and   
 $A=A_0$, $B=B_0$,  $\vp=\rhoGinf$.

\subsubsection{How to formulate   quantum time evolutions?}
\label{subsubsec:time}

We now discuss a crucial 
 problem of how to formulate quantum time evolutions.  
 It  appears   directly relevant  to physics;
 it is  not merely a matter of mathematical rigor. 
In  \cite{WO, wok}, 
    the local Gibbs state and  the local time-translation symmetry  
 are given  by  the 
 {\emph{same}} local Hamiltonian $\Hlam$. 
This assumption   
 is essential for the proof of
  \eqref{eq:WOK-temp-nai} \eqref{eq:WOK-gr-nai}.
However, to  establish non-existence of something completely, 
one should take   (infinitely many)  possibilities into account.   
 Different choices  of local Hamiltonians  
for a local Gibbs state and a local time-translation 
 on the same $\Lam$  may be possible; 
 there is no  reason to exclude them.
Furthermore,   the  same-local-Hamiltonian prescription 
 used in \cite{WO, wok} needs  justification,  
 because  the  true  quantum time evolution instantly evolves   
 local observables  to nonlocal ones, whereas 
 the cut-off time evolution on $\Lam$ used in \cite{WO, wok}  unnaturally 
  confines   local observables of  
  $\All$ in its slightly larger subsystem $\Allplus$ eternally   
  not allowing   them     to escape from the given  region.
In the following, we   estimate   the  difference between 
  the cut-off 
 and  infinite-volume time evolutions 
 for   short-range  quantum spin lattice models.  

\begin{prop}
\label{lem:new}
Let $\Al$ denote a quantum spin system  on the lattice $\Zmu$.
Suppose  that 
the time evolution 
 $\{\alpha_t\in\AUT,\;t\in \R\}$ is translation invariant, strongly continuous,  and  approximately inner.
Let 
$\{\altLam\in\AUT,\;t\in \R\}$ denote 
  the local Heisenberg time evolution 
 for $\Lam\Subset\Gamma$ as given in \eqref{eq:HeisenLam}.
 Let       
 $\{A_x=\tau_x(A_0); \; x\in \Gamma\}$, 
where $A_0$ is a local operator with its  support   centered at  
 the origin of $\Zmu$.
Let  $\varepsilon>0$ and $t_0>0$.  Then 
for sufficiently large $\Lam\Subset\Gamma$
 the following estimate holds for all $t\in [-t_0,\ t_0]$\rm{:}
\begin{align}
\label{eq:TIMEeps}
\left \Vert \alpha_t\bigl(\hatAlam\bigr)- \altLam\bigl(\hatAlam\bigr) \right \Vert
\le  \varepsilon.
\end{align}
\end{prop}

 \begin{proof}
For each $n\in \NN$
define 
\begin{align}
\label{eq:Lamn}
\Lamn:= \left\{(x_1, x_2, \cdots, x_\mu)\in \Zmu;\; 0\le |x_i|\le \frac{n}{2}
   \right\}
\Subset \Zmu. 
\end{align}
It is a  box region centered at 
the origin $0 \in\Zmu$  and  
  ${\rm{diam}}(\Lamn)=n$ or $n-1$. 
Let $\Lamxn:= \Lamn+x$, i.e. 
  the translation of $\Lamn$ by $x\in \Zmu$.
From the assumption \eqref{eq:APP}
we have  
\begin{equation}
\label{eq:Ninfinitedyn}
\alpha_t(A)=\lim_{n\to \infty}
\altLamn(A) 
 \quad  \text{for each}\  A\in \Al
  \ \text{and }\  t\in \R.
\end{equation}
By  \eqref{eq:hatAlam}  for any $\Lam\Subset\Gamma$
\begin{align}
\label{eq:estGWA}
\altLam\bigl(\hatAlam\bigr)= 
 \frac{1}{|\Lam|} \sum_{x\in\Lam} 
 \altLam\bigl(A_x \bigr).
\end{align}
 Since the time evolution under consideration 
 is  space  translation invariant,   
 for any fixed  $\varepsilon>0$ and $t_0>0$,  there 
 exists a constant $m(>r)\in\NN$ (that is independent of $x\in \Gamma$)  
 such that  the following estimate holds
\begin{equation}
\label{eq:appDYNGvsGEN}
\Bigl\Vert \alpha_t(A_x)-\altIx(A_x)
\Bigr \Vert < \varepsilon/2 
 \quad \text{for any} \ t\in [-t_0,\ t_0]\ \text{and}\ x\in \Gamma,
\end{equation}
where  $\Ix$ is  any  finite subset 
that includes the box region $\Lamxm$ centered at $x$:  
\begin{equation}
\label{eq:Ix}
\Ix\supset \Lamxm.
\end{equation}
We now take  a sufficiently large  $\Lam\Subset\Gamma$
 such that $\Lam\ni 0$. 
We  divide $\Lam$ into the following two complement regions: 
 \begin{align}
\label{eq:Lamdivide}
\Lamkikueps:=\{x\in \Lam;\; \Lamxm \subset \Lam\}, 
\quad \Lamhanpaeps:=\Lam\setminus\Lamkikueps, 
\end{align}
 where the subscript indicates  
$\varepsilon$-dependence, but 
 $t_0$ dependence is omitted as there is no fear of confusion.
Hence by the obvious inclusion  
$\Lam \supset \bigcup_{x\in \Lamkikueps} \Lamxm$, 
from  \eqref{eq:appDYNGvsGEN}
 \eqref{eq:Ix}
it follows that 
 \begin{equation}
\label{eq:appxhairu}
\Bigl\Vert \alpha_t(A_x)-\altLam(A_x)
\Bigr \Vert < \varepsilon/2 
 \quad \text{for any} \ t\in [-t_0,\ t_0] \ \text{and}\  x\in \Lamkikueps.
\end{equation}
Let   $a:=\Vert A_0 \Vert$.
By using  \eqref{eq:estGWA}
\eqref{eq:appxhairu} 
we obtain
\begin{align}
\label{eq:estSA}
&\left \Vert \alpha_t\bigl(\hatAlam\bigr)- \altLam\bigl(\hatAlam\bigr) \right 
\Vert=
 \frac{1}{|\Lam|}   \left \Vert \sum_{x\in\Lam} \left( 
\alpha_t\bigl(A_x\bigr)- 
\altLam\bigl(A_x \bigr) \right) \right \Vert
\nonumber \\
&\le  
 \frac{1}{|\Lam|}  \sum_{x\in\Lam} \left \Vert 
\alpha_t\bigl(A_x\bigr)- 
\altLam\bigl(A_x \bigr) \right \Vert
\nonumber \\
&=
\frac{1}{|\Lam|}  \sum_{x\in\Lamkikueps} \left \Vert 
\alpha_t\bigl(A_x\bigr)- 
\altLam\bigl(A_x \bigr) \right \Vert
\nonumber\\
&\quad +
\frac{1}{|\Lam|}  \sum_{x\in\Lamhanpaeps} \left \Vert 
\alpha_t\bigl(A_x\bigr)- 
\altLam\bigl(A_x \bigr) \right \Vert
\nonumber  \\
&\le 
\frac{1}{|\Lam|}  \sum_{x\in\Lamkikueps} \left \Vert 
\alpha_t\bigl(A_x\bigr)- 
\altLam\bigl(A_x \bigr) \right \Vert
+
\frac{1}{|\Lam|}  \sum_{x\in\Lamhanpaeps}  \Vert 
2 A_x \Vert
\nonumber  \\
&\le  
\frac{1}{|\Lam|}  \cdot |\Lamkikueps|\cdot  \frac{\varepsilon}{2} 
+
\frac{1}{|\Lam|}  |\Lamhanpaeps| \cdot 2a 
\nonumber  \\
&\le  
 \frac{\varepsilon}{2} + 2 a\frac{|\Lamhanpaeps|}{|\Lam|}  
\end{align}
By  \eqref{eq:Lamdivide}
 for each fixed $\varepsilon>0$ we have 
\begin{align}
\label{eq:vanHove}
\lim_{\Lam \rightsquigarrow \infty}
\frac{|\Lamhanpaeps|}{
|\Lam|}=0, 
\end{align}
where  the above  infinite-volume limit 
$\Lam \rightsquigarrow \infty$
is  the so called  van Hove limit, see Sec. 6.2.4 of \cite{BR}.
By \eqref{eq:estSA}
and \eqref{eq:vanHove}
  for sufficiently large $\Lam\Subset\Gamma$
 we obtain the estimate \eqref{eq:TIMEeps}. 
 \end{proof}

Lemma  \ref{lem:new}
 shows 
  that the difference 
 between these two different  time evolutions of  any local order parameter
   disappears in the infinite-volume limit.
The  following result follows.

\begin{prop}
\label{prop:LRO-kurabe}
Under the same assumption as 
 Lemma  \ref{lem:new}, 
the following identity holds{\text{:}}
\begin{equation}
\label{eq:tripleLIMITs}
\lim_{\LamiitoG}
\lim_{\LamitoG}
\rhoGlami\left(\altLamii\bigl( \hatAlamii \bigr) \hatBlamii  \right)
=\corweakABzero(t), 
\end{equation}
where the right-hand side  is given by the formula 
 $\corvpAB(t)$
\eqref{eq:corvpAB} 
 in  Section~\ref{subsubsec:LROCstar}
with  the approximately inner time evolution $\{\alpha_t\in\AUT,\;t\in \R\}$  as our $\cstar$-dynamics, 
 $A=A_0$, $B=B_0$, and $\vp=\rhoGinf$.
\end{prop}

\begin{remark} 
\label{rem:precise}
Proposition \ref{prop:LRO-kurabe} does not imply   that       
\begin{equation}
\label{eq:conjecLIMIT}
\lim_{\LamtoG} \corlocGibbsAB(t)
=\corweakABzero(t), \quad t\in \R.
\end{equation}
This  equation  seems  to be difficult to prove or disprove. 
We thus consider 
 that   Corollary \ref{coro:NoLRO}  based on  the $\cstar$-algebraic 
 LRO    is   independent 
 of  the  no-go statement of \cite{WO, wok} based on 
 the Griffiths-type LRO. 
\end{remark}

\subsection{On generality of  assumptions}
\label{subsec:Generality}
We have seen that 
the formulations of  ours  and   \cite{WO, wok}  
  are   different. 
Those  may  describe   different physics situations, and 
 thereby,   
 the meaning of  ``non-existence of  periodic 
  temporal correlation functions''   by us 
 and that of  \cite{WO, wok}  are not same, cf.
 Remark \ref{rem:precise}.
Nevertheless, 
we shall compare  our work with   \cite{WO, wok} 
in terms of  the generality of  assumptions
 in order to find the precise validity   
 of  these similar but different no-go statements.

\subsubsection{Peridoic and aperiodic crystals}
\label{subsubsec:Aperiodic}
Theorem \ref{thm:MAIN}
thoroughly excludes any 
 type of  quantum time crystals  
 such as  periodic space-time  crystals 
 as  in  \cite{LI}  and  also 
 aperiodic time crystals as in   \cite{QuasiA}
\cite{QuasiGMS} for  equilibrium states.
One may  imagine 
temporal orders  similar to  interfaces (domain walls) \cite{DOB72}
 or  turbulent  crystals (or chaotic crystals)  \cite{RU82}. 
 Those are negated by Theorem \ref{thm:MAIN} as well.
On the other hand, such inhomogeneous  quantum time crystals 
can  not   be  precluded  
  by   \cite{WO, wok}, as it   
 essentially requires    the  spatial homogeneity in its proof.

\subsubsection{Equilibrium states under consideration}
\label{subsubsec:Equilibrium-inclusion}
 The inclusion  of \eqref{eq:ACCUsmall}
  becomes identity 
  for  some general  
 translation invariant  classical spin lattice models
 as shown in  Theorem 6.63 \cite{Friedli}.
 However, for   
  quantum spin lattice models,  
 it is not known whether the  
 inclusion of \eqref{eq:ACCUsmall} is  strict  or not.
In  the method   
 of   \cite{KT-Finite} on which Watanabe-Oshikawa-Koma's argument relies,  
 the sequence 
 of symmetric local  Gibbs states   is taken. 
 However,  according to \cite{COQ}, there is an equilibrium state 
 that can not be obtained by 
  limits of  symmetric  local Gibbs states. 
For the boson system,  the inclusion  of \eqref{eq:ACCUsmall}
 is strict, as   there are  KMS states
 which are not  limiting Gibbs states \cite{SMED}.
Thus,  the set of translation invariant equilibrium states
  considered  in this review 
is  much larger than that considered  in \cite{WO, wok}.

\subsubsection{Is  the 
 Lieb-Robinson bound argument really necessary?}
\label{subsubsec:Models}
In  \cite{WO, wok} 
 the  Lieb-Robinson bound estimate \cite{LRbound} for local Hamiltonians 
is  used   in the derivation of   \eqref{eq:WOK-temp-nai}. 
 However,   is the Lieb-Robinson bound 
 essential  for the absence of genuine quantum time crystals?
It has been known \cite{BR} that  the  Lieb-Robinson  bound estimate
yields  approximately inner  $\cstar$-dynamics,  
which  is  our  important assumption.
On the other hand, the converse implication is not  true in general.  
 There are    $\cstar$-dynamics  
  not satisfying  the Lieb-Robinson  bound;
  examples   are given     by 
  quasi-free automorphisms  on the fermion lattice system \cite{ARCAR}.
 Proposition \ref{prop:INV},  
  Theorem \ref{thm:MAIN} and  
 Corollary \ref{coro:NoLRO}
 can be applied 
to  such long-range $\cstar$-dynamics,   
whereas the no-go statement 
of \cite{WO, wok} cannot. 
Let us mention  other 
long-range models \cite{BRUPEDRA} for   which 
 our results  are valid.

\subsubsection{Limitations of $\cstar$-algebraic approach}
\label{subsubsec:LIMIT}
So far we have emphasized wide generality of our results.
 We now  mention   restrictions 
  of our results based on  the 
 existence of   $\cstar$-dynamics.
 We note that concrete  examples of 
 $\cstar$-dynamics are rather  exceptional such as  
  short-range quantum spin lattice  models
  and  non-interacting
 quantum field models  \cite{BR}.
In general, construction of  
 $\cstar$-dynamics  (or $\wstar$-dynamics)
 for quantum-field models is  formidable. 
Some long-range  quantum spin lattice models
 do not have their $\cstar$-dynamics.
For example,  a  strong coupling  BCS  model does not 
 have  its  infinite-volume 
  time evolution  as  $\cstar$-dynamics;
it  only exists   in a state-dependent manner \cite{THIRBCS} 
\cite{BRUPEDRA-BCS}. 
On the other hand, 
the formulation of 
WOK  temporal correlation functions
 is more  flexible.
Kozin-Kyriienko \cite{KKlong} showed 
 that  some infinite-range Hamiltonian 
  generates  a non-trivial  periodic   
 WOK  temporal correlation function
 and they claimed that it is an example 
 of genuine quantum crystals. Obviously, 
 Corollary \ref{coro:NoLRO}   cannot be applied to
Kozin-Kyriienko's  model. 
 (The  feasibility of the  long-range Hamiltonian 
  and the highly entangled ground states  of 
 Kozin-Kyriienko's  model 
 has been  critically   argued in \cite{CommLRI}.)

\subsection{Is theory of relativity    relevant?}
\label{subsec:RELATIVITY}

Why are  quantum time crystals  impossible, whereas 
 spatial crystals in equilibrium  states are common?
In  the article  \cite{SCIAMEWIL}  to general audience,  
   Wilczek  \cite{WIL} 
 addresses the theory of relativity as his    motivation 
 to pursue  the above question.
 In  \cite{WO}the theory of relativity is also addressed.
 Although we cannot  
  specify  the  intension  of these authors, 
 we highlight the following seemingly relevant facts
 of the theory of relativity and  
 quantum equilibrium states. 
\begin{enumerate}   
\item Thermal equilibrium states in 
a fixed Lorentz-frame   can  be 
characterized by the KMS 
 condition  \cite{BU-JUG}. 
 The KMS states 
   violate the Lorentz-symmetry \cite{OJIMA}, while 
 they   always  preserve the time-translation symmetry 
  as shown in Proposition \ref{prop:INV}.

\item The spectrum condition 
 of local quantum 
 physics \cite{HAAG}  postulates that the spectrum of 
Hamiltonian and momentum operators on the Hilbert 
 space of   a vacuum state  
is  included in the forward-light cone in Lorentz space-time. 
(Here  the vacuum 
state  is not necessarily  Lorentz-invariant.)
The spectrum condition 
 forbids   crystalline  structure in {\emph{space}} 
 \cite{ARA64}.  See  also Theorem  4.6 
\cite{ARbook}  and  Theorem 3.2.4 \cite{HAAG}. 
Thus   relativistic  vacuum states  allow 
  crystalline  structure   {\emph{neither}} in the  space direction 
nor in the  time direction.
\end{enumerate}   

It looks that the theory of relativity
 will give   certain  restrictions upon possible crystal structure 
 on space-time.
From  a purely scientific perspective, 
we cannot find  a meaningful link between 
 the notion of genuine quantum crystals and the  theory of   relativity.

\subsection{On non-equilibrium  quantum time crystals}
\label{subsec:Validity}
In the final part  of \cite{wok},  
   quantum time crystals 
  by non-Gibbsian states 
are discussed.
In  the $\cstar$-algebraic language as well, 
 a general formula  of certain  non-equilibrium   
quantum time crystals can be given 
 putting aside their concrete  realization.
  It is  a  naive generalization of the notion of  
 SSB to stationary states 
   by using  
 the identification of factor states and pure phases as follows.
Suppose that  $\tilde \psi$ is  an   invariant state  
under the time-translation symmetry $\{\alpha_t\in\AUT,\;t\in \R\}$
 but  it is not an equilibrium state.
 Suppose that $\tilde \psi$ has the following specific factorial  
 decomposition:   For some $p>0$ 
\begin{equation}
\label{eq:GSSB}
\tilde \psi = \int_0^p dt   \,
\alpha_t^{\ast} \psi,\quad 
\psi \in  \FACState,  
\end{equation}
where $\psi$ is a factor state  
 breaking   the  time-translation symmetry 
$\{\alpha_t\in\AUT,\;t\in \R\}$
 but  invariant under its discrete  subgroup 
 $\{\alpha_t\in\AUT,\;t\in p \Z\}$. 
If   $\psi$ is a 
homogeneous  state, then 
by the equivalence of  the existence of non-trivial 
 $\cstar$-algebraic LRO and  that   of multiple phases  \cite{SEW2014}, 
the function  $\cortilpsiAA(t)$
 defined by the formula \eqref{eq:corvpAB}
 oscillates   periodically in time 
  for some local order parameter $A\in \Al$.

\section{Discussion}
\label{sec:DIS}
Based on  the  KMS condition,  
 we  gave   no-go statements    
of genuine quantum time crystals in the $\cstar$-algebraic formulation. 
The above   no-go statements 
 based on $\cstar$-dynamics  have  wider generality than \cite{WO, wok}
 in several  points, although they  were 
 essentially obtained  in 1970s.
 Our    viewpoint 
upon the notion of 
 genuine quantum time crystals 
and the   no-go statement 
 used  in the physics literature
  is  contrasting to 
   \cite{HASACHA2} \cite{KHE} \cite{SACHA}.

Now let us come back to the following fundamental 
 problem addressed in the discussion of \cite{WO}:
Why are  quantum time crystals  impossible for  equilibrium states, whereas 
 spatial crystals in equilibrium  states are common?
In Section \ref{subsec:RELATIVITY}, 
 we emphasized    that 
 the theory of relativity
 is  irrelevant to the above question.
 We consider   that the impossibility  of
 genuine quantum  time crystals 
 is due to  rigid   stability  of  quantum equilibrium states.  
  The KMS condition  is known to imply and to be  implied by
  several  characterizations of equilibrium,   see  \cite{BR}.
Thus,  Theorem \ref{thm:MAIN} given in terms of the 
 KMS condition  can be rephrased   as follows:

\medskip

$\bullet$ The variational 
 principle for equilibrium states 
  forbids 
the existence  of genuine quantum time crystals.

\medskip

$\bullet$  The  passivity   by Pusz-Woronowicz 
forbids the existence  of genuine quantum time crystals. 

\medskip

This review  has narrowed 
  the possibility  of genuine quantum time crystals. 
But as we have noted, 
 there are   many quantum  models  
 that can not be formulated in  the $\cstar$-algebraic
 formulation. 
Although  we consider that 
 our no-go statements have wider genarality beyond the  $\cstar$-algebraic
 formulation,  these   should not  be 
 applied to such models unless there is a rigorous proof.

\section*{Acknowledgments}
I   thank my  colleagues 
 for  conversations     on  the issue.
 Prof.  Araki    told me about his   scientific interaction 
  with Prof. Ryogo Kubo  about the  KMS condition. 
I acknowledge  Kakenhi (grant no.
21K03290) for the financial support.
 I thank 
  Institute of Mathematics for Industry, Joint Usage/Research Center 
 in Kyushu University for the financial support to participate in   
   FY2022 Workshop (I) ``Theory and experiment for time, quantum measurement and semiclassical approximation -Interface between classical and quantum Theory-". 

\end{document}